\documentclass[10pt]{article}
\usepackage{amsmath,amssymb,bbm,comment}
\usepackage{amsthm}
\usepackage{mathrsfs}

\def \FFF {\mathfrak{F}}
\def \BB {\mathscr{B}}

\def \DD {\mathcal{D}}

\def \HH {\mathcal{H}}

\def \NN {\mathcal{N}}

\def \CCC {\mathbb{C}}

\def \equals {\ = \ }

\def \Tr {\mathrm{Tr}}

\def \bn {{\mathfrak n}}

\def \lu {}

\def \Cor {\mathfrak{C}}

\newtheorem{lemma}{Lemma}
\newtheorem{proposition}{Proposition}
\newtheorem{corollary}{Corollary}
\newtheorem{remark}{Remark}

\newtheorem{definition}{Definition}

\newtheorem*{proposition*}{Proposition}

\author{
Alex D. Gottlieb\footnote{Wolfgang Pauli Institute c/o Fakult\"at f. Mathematik, Universit\"at Wien, Oskar-Morgensternplatz 1, 1090 Vienna, Austria} 
\ and Norbert J. Mauser\footnotemark[1]
}

\title{Nonfreeness and related functionals for measuring correlation in many-fermion states}

\begin{document}

\maketitle

\begin{abstract}




This article is a brief review of ``nonfreeness" and related measures of ``correlation" for many-fermion systems.  

The many-fermion states we deem ``uncorrelated" are the gauge-invariant quasi-free states.  Uncorrelated states of systems of finitely many fermions we call simply ``free" states.  
Slater determinant states are free; all other free states are ``substates" of Slater determinant states or limits of such.  

The nonfreeness of a many-fermion state equals the minimum of its entropy relative to all free states.
Correlation functionals closely related to nonfreeness can be defined in terms of R\'enyi entropies; nonfreeness is the one that uses  Shannon entropy.   
These correlation functionals all share desirable additivity and monotonicity properties,
but nonfreeness has some additional attractive properties.

\end{abstract}



\section{Introduction}

``Nonfreeness" is an entropy functional of states of many-electron systems.  It was introduced as a ``measure of electron correlation" \cite{GottliebMauser2007,GottliebMauserArchived} 
that is purely a functional of the many-electron state, depending only on the structure of the state and not upon the physical circumstances attending it, e.g., the Hamiltonian operator for the system \cite{correlation}.  

By definition, the nonfreeness of a many-fermion state is its entropy relative to the unique gauge-invariant quasi-free state with the same $1$-particle density matrix ($1$-pdm).       
Gauge-invariant quasi-free (GIQF) states have $0$ nonfreeness by definition, but the nonfreeness of any other many-fermion state is positive, possibly infinite.  
Slater determinant states of $n$-fermions and ``Fermi sea" states of infinitely many fermions are GIQF, as are 
restrictions of such states.
Conversely, any GIQF state can be represented as 
restriction of a Slater determinant or Fermi sea state.  
These are the states we deem ``uncorrelated." 

In this article we shall mainly consider normal states of finite systems of fermions, identifying such states with the density operators that represent them on a fermion Fock space.  
Among such states we focus on those that have finite expected particle number.  
The GIQF states of finite average particle number we call simply ``free" states.

For pure $n$-fermion states, 
the nonfreeness functional coincides with ``particle-hole symmetric correlation entropy" \cite{Gori-GiorgiZiesche}.   
Particle-hole symmetric correlation entropy has been used to quantify electron correlation in the uniform electron gas \cite{Gori-GiorgiZiesche} and in short linear chains undergoing a Mott transition \cite{Bendazzoli...}. 
Particle-hole symmetric correlation entropy is defined only for pure states; nonfreeness is an extension of that functional to the domain of mixed many-fermion states, i.e., states that can be represented by density operators on the fermion Fock space.   

A correlation functional for mixed states can be useful even if the system of interest is one of exactly $n$ fermions in a pure state, 
because open subsystems of the system of interest are typically in a mixed states, containing a random number of particles.  
Consider, for example, a system of fermions on a lattice.  
The fermions that occupy a given site or block of sites constitute a subsystem that is typically in a mixed state, and the von Neumann entropy of that local state can reflect physical properties such as quantum phase transitions \cite{VidalLatorreRicoKitaev,GuDengLiLin,LarssonJohannesson}.   
Indeed, nonfreeness has been used to quantify local correlation in a realistic 
tight-binding model of a transition metal oxide heterostructure  \cite{MauserHeld}.    

The state of a many-fermion system determines the states of all its subsystems (e.g., local states in an extended system).  
The induced state of a subsystem may be called a  ``restriction" \cite{HainzlLewinSolovej} or ``localization" \cite{LewinNamRougerie} of the state; we call it a ``substate."  
Nonfreeness is monotone with respect to restriction of states:  the nonfreeness of a substate is less than or equal to the nonfreeness of the state from which it is derived \cite{GottliebMauser2007}.   
Also, nonfreeness is additive over independent subsystems:  when a many-fermion state is a product of statistically independent substates, its nonfreeness is the sum of its substates' nonfreeness \cite{GottliebMauser2007}.

The monotonicity and additivity properties of nonfreeness derive from its definition as a relative entropy.  
Correlation functionals closely related to nonfreeness can be defined using R\'enyi divergences instead of relative entropy.
R\'enyi divergences also enjoy the properties of additivity and monotonicity, and so do the correlation functionals defined in terms of them.   
Indeed, the ``new measure of electron correlation" that we proposed in Ref.~\cite{GottliebMauser2005} is of this type.  
However, within this class of correlation functionals, the nonfreeness functional has a couple of additional attractive properties, presently to be stated precisely.

Suppose $\Delta$ is a density operator on a fermion Fock space that represents a state of finite average particle number, 
and let $\Gamma_\Delta$ denote the density operator of the unique free state with the same $1$-pdm as $\Delta$.   
The nonfreeness of $\Delta$, or of the state it represents, 
is defined to be $S(\Delta \|\Gamma_\Delta)$, the entropy of $\Delta$ relative to $\Gamma_\Delta$. 

The nonfreeness of $\Delta$ is given by the following simple formula, provided its von Neumann entropy $S(\Delta) = -\Tr (\Delta \log\Delta)$ is finite:
\begin{eqnarray}
\label{simple formula}
\lefteqn{ S(\Delta \| \Gamma_{\Delta}) \equals  S(\Gamma_{\Delta}) - S(\Delta) }\nonumber \\
        & = &   - \sum p_j \log p_j - \sum (1-p_j) \log(1-p_j)  \ - \ S(\Delta) \qquad
\end{eqnarray} 
where the $p_j$ denote the eigenvalues of the $1$-pdm of $\Delta$, its natural occupation numbers. 

In any case, the nonfreeness of a many-fermion state is equal to the minimum of its entropy relative to all free states: 
\begin{equation}
\label{minimum property}
 S(\Delta \|\Gamma_\Delta) \equals \min \big\{S(\Delta\|\Gamma ) : \ \Gamma  \ \hbox{is free} \big\}\ .
\end{equation}
Moreover, if the minimum in (\ref{minimum property}) is finite, then $\Gamma = \Gamma_\Delta$ is the unique minimizer of $S(\Delta\|\Gamma )$ over all free states $\Gamma$.  

The nonfreeness of a many-fermion state $\Delta$ is its entropy relative to the free ``reference" state $\Gamma_\Delta$.  
Other authors have essayed similar relative entropy measures of correlation strength, using various other uncorrelated reference states chosen {\it ad hoc} on physical grounds \cite{Byczuk,ByczukErrata}. 
They proposed to use judicious choices of ``physically well-known uncorrelated states" $\Gamma$ as reference states, avowedly because  \cite{Byczuk footnote} they did not know which choice of $\Gamma$ minimizes $S(\Delta\|\Gamma)$.   
Shortly afterward, Held and Mauser \cite{MauserHeld} pointed out that the minimizer is $\Gamma_{\Delta}$, and 
and argued that (\ref{minimum property}) means that other choices of $\Gamma$ necessarily overestimate correlation.    In this review we will present a very thorough proof of (\ref{minimum property}).  


\[
\star
\]

The rest of this article is organized as follows:

Section~\ref{Density operators on Fock space} presents the notation and terminology required for reading 
Sections~\ref{Free states} - \ref{Special properties of nonfreeness}.   

In Section~\ref{Free states} we define free states.  Prop.~\ref{free at last} there asserts that free states are substates of Slater determinant states or limits of such states. 

In Section~\ref{Relative entropy correlation functionals} we discuss correlation functionals that are closely related to nonfreeness, focusing on properties they share.  

In Section~\ref{Special properties of nonfreeness} we review special properties of nonfreeness.   
The simple formula (\ref{simple formula}) for nonfreeness is developed in Prop.~\ref{max ent prop} and its corollary; and the minimum property (\ref{minimum property}) of nonfreeness is proven in Prop.~\ref{superproposition}.  

In order not to impede the review while keeping the article as a whole fairly self-contained, 
many of the technical details and some of the proofs have been removed to Section~\ref{Appendices}, which effectively consists of nine appendices.


\section{Density operators on Fock space}
\label{Density operators on Fock space}

Let $\HH$ denote a Hilbert space, the $1$-particle Hilbert space.  Unit vectors in $\HH$ are called ``orbitals."    Although the notation we will use suggests that $\HH$ has countably infinite Hilbert dimension, this is not required; everything works for $\HH$ of any dimension.  

The Hilbert space for finite systems of fermions in $\HH$ is the fermion Fock space over $\HH$, which we shall denote by  $\FFF(\HH)$ or simply $\FFF$.
Let $\hat{a}^*(f)$ and $\hat{a}(f)$ denote the creation and annihilation operators for $f \in \HH$, 
defined as bounded operators on the Fock space \cite{BratteliRobinson}.

An $n$-fermion ``Slater determinant" wave function can be identified with a Fock space vector 
\begin{equation}
\label{Slater determinant in Fock space}
       |\Phi\rangle  \equals \hat{a}^*(f_1)\hat{a}^*(f_2) \cdots \hat{a}^*(f_n)|\Omega \rangle
\end{equation}
where $f_1,\ldots,f_n$ are orthonormal orbitals in $\HH$.  
We think of the vacuum vector $|\Omega \rangle$ as a $0$-particle Slater determinant wave function.  

In this article, we are going to focus on states of many-fermion systems that can be represented by density operators on the Fock space, 
especially those that represent states of finite average particle number.  

Let $\Delta$ be a density operator on $\FFF(\HH)$.  
The ``$1$-particle density matrix" or ``$1$-pdm" of $\Delta$ is the bounded operator $\gamma_\Delta$ on $\HH$ such that 
\begin{equation}
\label{defining property of $1$-pdm}
    \langle g    | \gamma_\Delta f   \rangle \equals \Tr\big(\Delta \hat{a}^*(f)\hat{a}(g) \big)
\end{equation}
for all $f,g \in \HH$.  
The preceding formula implies that $\gamma_{\Delta}$ is a positive-semidefinite contraction.  Its eigenvectors are called ``natural orbitals" of $\Delta$, 
and the corresponding eigenvalues are the ``natural occupation numbers" of $\Delta$.  

If $h \in \HH$ is a unit vector,  the diagonal matrix element $ \langle h | \gamma_\Delta h  \rangle$ of $\gamma_\Delta $ 
is the probability that the orbital $h$ is occupied when the system is in the state represented by $\Delta$.  This is because 
the operator $\hat{a}^*(h)\hat{a}(h)$ corresponds to the physical observable of orbital $h$'s occupation, which takes the values $0$ and $1$.
The expected value of this observable, the probability that orbital $h$ is occupied, 
is therefore $\Tr\big( \Delta \hat{a}(h)^*\hat{a}(h) \big)$, and this equals $\langle h|\gamma_{\Delta} h \rangle $ by definition.  

We are especially interested in the class of density operators on $\FFF = \FFF(\HH)$ that represent states of finite average particle number.     
We shall denote this class by $\DD(\FFF)$.   
If $\Delta$ is a density operator on $\FFF$, then the average number of particles in the state represented by $\Delta$ equals the trace of its $1$-pdm.  
Thus, a density operator $\Delta$ belongs to $\DD(\FFF)$ if and only if $\Tr (\gamma_\Delta) < \infty$.  
Note that $\DD(\FFF)$ contains all the Slater determinant states $|\Phi\rangle\!\langle \Phi |$ where $|\Phi\rangle$ is a Slater determinant wave function (\ref{Slater determinant in Fock space})  in $\FFF$.

\section{Free states}
\label{Free states}

In this section we shall define and discuss the kind of many-fermion states we consider uncorrelated.  
 
We restrict our attention to states that are represented by density operators on a fermion Fock space, and which have finite average particle number. 
Among such states, we consider Slater determinant states to be uncorrelated, as well as any state that can be represented as a ``substate" or restriction of a Slater determinant state.  
A state can be represented as a substate of a Slater determinant state if and only if it is ``gauge-invariant quasi-free" \cite{AlickiFannes}  and its $1$-pdm has finite rank.  
We want limits of free states to be free, too.  
The smallest class of states containing all substates of Slater determinant states and limits of such is the class of ``free" states, which we define as follows:

\begin{definition} 
\label{free def}
A density operator $\Gamma$ on a fermion Fock space is ``free" if it represents a gauge-invariant quasi-free state and its $1$-pdm has finite trace.
\end{definition}
\noindent
Quasi-free states with finite expected particle number are called ``generalized Hartree-Fock" states in Ref.~\cite{BachLiebSolovej}.  
Accordingly, free states are gauge-invariant generalized Hartree-Fock states.  

Gauge-invariant quasi-free states are discussed in Section~\ref{Gauge-invariant quasi-free states}.
In Section~\ref{Free states appendix} we will prove that all free states are limits of substates of Slater determinant states:

\begin{proposition}  
\label{free at last}
A density operator on a fermion Fock space is free if and only if (i) its $1$-pdm has finite trace, and (ii) it is a limit in trace norm of density operators that represent substates of Slater determinant states.
\end{proposition} 

Free states whose natural occupation numbers are all strictly positive and less than $1$ have the form of (grand canonical) Gibbs states for noninteracting fermions \cite{GottliebMauserArchived}.   
Let $\hat{a}_1,\hat{a}_2,\ldots$ denote the fermionic annihilation operators associated to a complete orthonormal system of reference orbitals, 
so that $\hat{n}_i = \hat{a}^*_i \hat{a}_i$ represents the observable of ``the number of fermions in the $i^{th}$ orbital" (either $0$ or $1$).  
Any density operator $\Gamma$ of the Gibbs form 
\begin{equation}
      \Gamma \ \propto \  \exp\big( - \sum \lambda_i \hat{n}_i\big)
 \label{Gibbs form}
\end{equation}
is free.  
Formula (\ref{Gibbs form}) defines a density operator if and only if $\sum e^{-\lambda_i} < \infty$ 
because the trace of the operator on the right-hand side of the formula equals $\prod\big(1+e^{-\lambda_i} \big)$.  
The reference orbitals of the density operator $\Gamma$ defined by (\ref{Gibbs form}) are its natural orbitals.  
Its natural occupation numbers, the average values of the observables $\hat{n}_i$, are 
$ p_i = e^{-\lambda_i}/(1+e^{-\lambda_i}) $.
For later use we note here that $\log\Gamma$ is the  ``quadratic Hamiltonian" \cite{BachLiebSolovej} operator 
\begin{equation}
\log \Gamma  
      \equals 
      \sum_i \big( \log(p_i) \hat{a}^*_i \hat{a}_i +  \log(1-p_i) \hat{a}_i \hat{a}^*_i \big) \ . 
 \label{quadratic}
\end{equation}

Free states are characterized by statistically independent occupation of their natural orbitals.  
In Section~\ref{Free states appendix} we show that a density operator $\Gamma$ on $\FFF(\HH)$ is free if and only if orthogonal natural orbitals are occupied independently of one another.   
For example, in the free Gibbs state (\ref{Gibbs form}) the expected value of the occupation observables $\hat{n}_i$ and $\hat{n}_j$ are $p_i$ and $p_j$, respectively, while the expected value of  $\hat{n}_i\hat{n}_j$, i.e., the probability that the $i^{th}$ and $j^{th}$ orbitals are both occupied, equals $p_ip_j$ (assuming $i \ne j$).

In order to define nonfreeness and related correlation functionals, we will require the following well-known fact \cite{BachLiebSolovej,AlickiFannes}, which we will also prove in Section~\ref{Free states appendix}.

\begin{remark}
\label{free existence remark}
Suppose $Q:\HH \longrightarrow \HH$  is a positive-semidefinite contraction operator with finite trace.  
Then there exists a unique free density operator on $\FFF(\HH)$ with $1$-pdm $Q$.  
\end{remark}

The von Neumann entropy of a free state is a simple function of its natural occupation numbers $p_i$.  
The following formula can be established using Proposition~\ref{structural} in Section~\ref{Free states appendix}.
  
\begin{remark}
\label{quasifree entropy remark}
 If  $\Gamma$ is a free state with natural occupation numbers $p_i$, then its von Neumann entropy is 
\[
   S(\Gamma) \ = \ -\sum_i p_i \log p_i -\sum_i (1-p_i) \log (1-p_i) \ .  
\]  
\end{remark}


\section{Relative entropy correlation functionals}
\label{Relative entropy correlation functionals}

Recall that $\DD(\FFF)$ denotes the set of density operators on the fermion Fock space $\FFF = \FFF(\HH)$ that represent states of finite average particle number.  
The $1$-pdm of a density operator $\Delta \in \DD(\FFF)$ is a positive-semidefinite contraction operator on $\HH$ with finite trace.  
By Remark~\ref{free existence remark}, there exists a unique free density operator with the same $1$-pdm as $\Delta$.
We denote this free density operator by $\Gamma_\Delta$.   
In other words, $\Gamma_\Delta$ denotes the unique free density operator such that $\gamma_{\Gamma_\Delta} = \gamma_\Delta$ (with the notation defined in formula (\ref{defining property of $1$-pdm})).  
 
The nonfreeness $\mathfrak{C}(\Delta)$ of $\Delta$ is defined to be the entropy of $\Delta$ relative to $\Gamma_\Delta$, that is, 
\begin{equation}
\label{nonfreeness intro def} 
\mathfrak{C}(\Delta) \equals  S(\Delta \|\Gamma_\Delta) \ .
\end{equation} 
This equals $\Tr(\Delta \log\Delta) - \Tr(\Delta \log\Gamma_{\Delta})$ provided that $\Tr(\Delta \log\Delta) > -\infty$.  

Correlation functionals closely related to nonfreeness are obtained by using other ``divergences" instead of the relative entropy to compare the states $\Delta$ and $\Gamma_\Delta$.  Using a divergence that enjoys the properties of additivity and monotonicity will yield a correlation functional with those properties.  
We have in mind the R\'enyi divergences 
\[
D_\alpha(\rho \| \sigma) \equals \frac{1}{\alpha - 1} \log \Tr (\rho^\alpha \sigma^{1-\alpha})
\]
for $0 < \alpha \le 2$ and the ``sandwiched" relative R\'enyi entropies \cite{Mueller-Lennert...,WildeWinterYang,FrankLieb} 
\[
\widetilde{D}_\alpha(\rho \| \sigma)  \equals \frac{1}{\alpha - 1} \log \Tr \Big(\big(\sigma^{\frac{1-\alpha}{2\alpha}} \rho \sigma^{\frac{1-\alpha}{2\alpha}}\big)^\alpha\Big)
\]
for $\alpha \ge \tfrac12$.   The divergences $D_1$ and $\widetilde{D}_1$ are defined by taking limits $\alpha \longrightarrow 1$ and both equal the relative entropy $S(\rho \| \sigma)$.

For values of $\alpha$ in the appropriate ranges, the correlation functionals  
\begin{eqnarray*}
\mathfrak{C}_\alpha( \Delta )  &=&  D_\alpha(\Delta \|\Gamma_{\Delta} ) \\
\widetilde{\mathfrak{C}}_\alpha( \Delta ) &=&   \widetilde{D}_\alpha(\Delta \|\Gamma_{\Delta} )  
\end{eqnarray*} 
all share the following properties with the nonfreeness functional $\mathfrak{C} = \mathfrak{C}_1 = \widetilde{\mathfrak{C}}_1$ :
\begin{trivlist}
\item{(i)} 
they take only non-negative values, possibly $+\infty$,
\item{(ii)}
they assign the value $0$ to all Slater determinant states,
\item{(iii)}
they are monotone with respect to restriction of states,
\item{(iv)} 
they are additive over independent subsystems, and 
\item{(v)}
they are invariant under changes of the $1$-particle basis.
\end{trivlist}

The sandwiched relative R\'enyi entropy $\widetilde{D}_{1/2}$ equals twice the negative logarithm of ``fidelity," 
and the corresponding correlation functional $\widetilde{\mathfrak{C}}_{1/2}$ is the ``new measure" of correlation we proposed in Ref.~\cite{GottliebMauser2005}.


\section{Special properties of nonfreeness}
\label{Special properties of nonfreeness}

Due to its definition in terms of von Neumann entropy, nonfreeness has some intuitively appealing properties 
that the other relative-entropy-type correlation functionals do not share.

The nonfreeness of a many-fermion density operator $\Delta$ has been defined to be its entropy relative to the 
associated free state $\Gamma_\Delta$.  Prop.~\ref{max ent prop} states that the nonfreeness of $\Delta$ equals the difference between the von Neumann entropies of $\Gamma_\Delta$ and $\Delta$.  
The nonfreeness $\Cor(\Delta)$ may be defined without direct reference to $\Gamma_\Delta$, by  
$$ \Cor(\Delta) \equals \min \big\{S(\Delta\|\Gamma ) : \ \Gamma  \ \hbox{is free} \big\} $$
because  the minimum relative entropy is actually attained at $\Gamma=\Gamma_{\Delta}$, as shown in Prop.~\ref{superproposition} below.

\subsection{Simple formulas for nonfreeness}
\label{Simple formulas for nonfreeness}

Recall that $\DD(\FFF)$ denotes the set of density operators on the fermion Fock space $\FFF$ that represent states of finite average particle number.  

\begin{lemma} 
\label{good lemma}
Suppose $\Delta \in \DD(\FFF)$ and let $\Gamma_{\Delta}$ denote the unique free state that has the same $1$-pdm as $\Delta$.
If $\Gamma$ is free then 
\begin{equation}
\label{conclusion}  -  \Tr(\Delta \log\Gamma )   \equals - \Tr(\Gamma_\Delta \log\Gamma )\ . \
\end{equation}

\end{lemma}

\begin{proof}  
We prove this here for the case where $\Gamma$ is a free Gibbs state, i.e., when all natural occupation numbers $p_i$ of $\Gamma$ are strictly between $0$ and $1$.  
The proof is simple in this case because the operator $\log \Gamma$ is then quadratic in the creators and annihilators, while $\Delta$ and $\Gamma_\Delta$, having the same $1$-pdm, assign the same expectations to all such operators.  The general case where some of the $p_i$ may equal $0$ or $1$ requires some care and is handled in Section~\ref{Proof of Lemma 1}.

Suppose 
$\log \Gamma$ is the quadratic Hamiltonian operator 
(\ref{quadratic}).  
By the defining property (\ref{defining property of $1$-pdm}) of the  $1$-pdm $\gamma_\Delta$, 
\begin{eqnarray*}
- \Tr (\Delta \log \Gamma ) & = & 
- \sum_i \big( \log(p_i) \Tr (\Delta \hat{a}^*_i \hat{a}_i) -  \log(1-p_i) \Tr (\Delta \hat{a}_i \hat{a}^*_i)  \big)
\\
& = & 
- \sum_i \big( \log(p_i) \langle h_i | \gamma_\Delta h_i \rangle -  \log(1-p_i) (1 -  \langle h_i | \gamma_\Delta f_i \rangle ) \big).
\end{eqnarray*}
Since $\gamma_\Delta$ is also the $1$-pdm of $\Gamma_\Delta$, the conclusion (\ref{conclusion}) follows.
\end{proof}  

When the von Neumann entropy $S(\Delta) = -\Tr (\Delta \log\Delta)$ is finite 
we may use the formula
\begin{equation}
\label{we may use}
S(\Delta \| \Gamma_{\Delta}) \equals - \Tr(\Delta \log\Gamma_{\Delta}) - S(\Delta)
\end{equation}
for the relative entropy.  This leads to simple formulas for nonfreeness.

\begin{proposition}
\label{max ent prop}
Suppose $\Delta \in \DD(\FFF)$ satisfies $S(\Delta) < \infty$.
Let $\Gamma_{\Delta}$ denote the unique free density operator with the same $1$-pdm as $\Delta$.   
Then, 
\begin{equation}  
\label{basic formula}
     S(\Delta\|\Gamma_{\Delta}) \equals S(\Gamma_\Delta)-S(\Delta) \ .
\end{equation}
\end{proposition}

\begin{proof}
By Lemma~\ref{good lemma}, 
\[
     -\Tr( \Delta \log \Gamma_{\Delta} )  \equals  -\Tr( \Gamma_{\Delta} \log \Gamma_{\Delta} ) \equals 
    S(\Gamma_\Delta ).
\]
Substituting $S(\Gamma_\Delta )$ for $-\Tr( \Delta \log \Gamma_{\Delta} )$ in equation (\ref{we may use}) yields (\ref{basic formula}).
\end{proof}

By Remark~\ref{quasifree entropy remark}, the von Neumann entropy of the free state $\Gamma_\Delta$ is a function of its natural occupation numbers.  
The natural occupation numbers of $\Gamma_{\Delta}$ are the same as those of $\Delta$, since they have the same $1$-pdm; therefore, using (\ref{basic formula}) we obtain the following simple formula for nonfreeness:
\begin{corollary}
\label{explicit nonfreeness formula}
Suppose $\Delta \in \DD(\FFF)$ is a density operator on the fermion Fock space, and let $p_1,p_2,\ldots$ denote the eigenvalues of its 1-pdm.
If $S(\Delta) < \infty$ then 
\[  
       S( \Delta \| \Gamma_{\Delta} )  \equals - \sum p_j \log p_j - \sum (1-p_j) \log(1-p_j)  - S(\Delta)\ .
\]
\end{corollary}

\subsection{Nonfreeness as relative entropy mimimum}
\label{Nonfreeness as relative entropy mimimum}

The nonfreeness $\mathfrak{C}(\Delta)$ of a many-fermion state $\Delta$ is equal to the minimum of its entropy relative to all free reference states.  
To prove this we will use the inequality 
\begin{equation}
\label{an entropy inequality}
  S( A ) \ \le\  -\Tr(A \log B) 
\end{equation}
for two density operators on the same Hilbert space.  
In case $S(A) < \infty$, then 
\begin{equation}
\label{if entropy is finite}
S( A \| B) = -  \Tr ( A \log B ) - S(A)
\end{equation}
and (\ref{an entropy inequality}) follows immediately from the fact that $S( A \| B) \ge 0$.  

\begin{proposition}
\label{superproposition}
Suppose $\Delta \in \DD(\FFF)$ and let $\Gamma_{\Delta}$ denote the unique free density operator with the same $1$-pdm as $\Delta$.   
Then, for all free density operators $\Gamma $,
\begin{equation} 
\label{superprop-equation}
     S(\Delta\|\Gamma_{\Delta}) \ + \ S(\Gamma_{\Delta}\|\Gamma ) \ = \ S(\Delta\|\Gamma  ), 
\end{equation}
and therefore 
\begin{equation}  
\label{superprop-inequality}
     S(\Delta\|\Gamma_{\Delta}) \ \le \ S(\Delta\|\Gamma  )
\end{equation}
with equality only if $S(\Delta\|\Gamma_\Delta) = \infty$ or $\Gamma=\Gamma_\Delta$. 
\end{proposition}

\begin{remark}
The analog of Prop.~\ref{superproposition} for other R\'enyi divergences would be false. 
 That is, if $\alpha \ne 1$ then $\Gamma_{\Delta}$ need not minimize $D_\alpha(\Delta \|\Gamma )$ or $\widetilde{D}_\alpha(\Delta \|\Gamma )$.   
 
 For example, let $\HH = \mathrm{span}\big\{| \uparrow\  \rangle,| \downarrow\  \rangle \big\}$ and let $\Delta$ denote the density operator that is entirely supported on the $1$-particle component of $\FFF(\HH)$, where it equals $\frac23 | \uparrow \ \rangle\!\langle \ \uparrow  | + \tfrac13 | \downarrow \ \rangle\!\langle \ \downarrow |$.   Then the minimum of $D_\alpha(\Delta \|\Gamma)$ or $\widetilde{D}_\alpha(\Delta \|\Gamma )$ is not attained at $\Gamma_\Delta$ unless $\alpha = 1$.
\end{remark}

\begin{proof}
We first prove Prop.~\ref{superproposition} under the assumption that $S(\Delta) < \infty$, which allows us to use formula (\ref{if entropy is finite}).  
Then we will relieve the assumption that $S(\Delta) < \infty$ by using the martingale property of relative entropy \cite{OhyaPetz}. 

Suppose that $S(\Delta) < \infty$.  Then 
\[
      S( \Delta \| \Gamma ) \ = \   -  \Tr ( \Delta \log \Gamma ) - S(\Delta)\ .
\]
 If $\Gamma$ is free, then 
\begin{eqnarray*}
S(\Delta \| \Gamma)  & = &  -\Tr(\Delta \log \Gamma) - S(\Delta) \equals  -\Tr( \Gamma_{\Delta} \log \Gamma ) - S(\Delta) \\
    & \ge & S(\Gamma_{\Delta}) - S(\Delta) \equals  -\Tr( \Delta \log \Gamma_{\Delta} ) - S(\Delta) \equals S(\Delta \| \Gamma_{\Delta} )  \ .
\end{eqnarray*}
The first and last equalities hold because $S(\Delta) < \infty$; the next-to-first and next-to-last equalities hold by  
Lemma~\ref{good lemma}; the inequality holds by (\ref{an entropy inequality}).   
This establishes (\ref{superprop-inequality}) when $S(\Delta) < \infty$.

If $S(\Delta \| \Gamma_{\Delta} ) = \infty$, then also $S(\Delta \| \Gamma)=  \infty$, as we have just established, and equation (\ref{superprop-equation}) holds trivially.   On the other hand, if $S(\Delta \| \Gamma_{\Delta} )$ and $S(\Delta)$ are both finite, then Prop.~\ref{max ent prop} implies that $S(\Gamma_\Delta )  < \infty$ 
and $S(\Delta\|\Gamma_{\Delta}) = S(\Gamma_\Delta)-S(\Delta)$.  By Lemma~\ref{good lemma}, 
\begin{eqnarray}
 S(\Delta \| \Gamma) - S(\Delta \| \Gamma_{\Delta} )
  & = & 
  -\Tr( \Delta \log \Gamma ) - S( \Gamma_{\Delta}  ) 
  \nonumber \\
  & = & 
      -\Tr( \Gamma_{\Delta} \log \Gamma ) - S( \Gamma_{\Delta}  ) 
      \equals 
  S(\Gamma_{\Delta} \| \Gamma), 
  \nonumber
 \end{eqnarray}
which is equivalent to equation (\ref{superprop-equation}).   Thus, equation (\ref{superprop-equation}) holds  
whether or not $S(\Delta \| \Gamma_{\Delta} )$ is finite, provided $S(\Delta ) < \infty$.  

Now assume that $S(\Delta) = \infty$.  The symbol $\BB(\mathcal{X})$ in the sequel denotes the algebra of bounded operators on a Hilbert space $\mathcal{X}$.

Consider an increasing sequence of finite-rank projectors $P_n$ on $\HH$ that converges strongly to the identity, and let $\HH_n$ denote the range of $P_n$.  
The finite-dimensional von Neumann algebras $\BB(\FFF(\HH_n))$ can be embedded into  $\BB(\FFF(\HH))$ as subalgebras, which we denote here by $\BB_n$.     
Let $\Delta_n$, $(\Gamma_{\Delta})_n$, and $\Gamma_n$ denote the density operators on $\FFF(\HH_n)$ that represent 
the restrictions of $\Delta$, $\Gamma_\Delta$, and $\Gamma$ to 
the corresponding substates delimited by $\HH_n$ (as defined in Section~\ref{Substates of many-fermion states}).
The density operator $(\Gamma_{\Delta})_n$ is free because it is a substate of a free state 
(see Section~\ref{Free states appendix}) 
 and it has the same $1$-pdm as $\Delta_n$, whence $(\Gamma_{\Delta})_n = \Gamma_{\Delta_n}$.  
Since $\Delta_n$ is a density operator on a finite-dimensional space, it has finite von Neumann entropy, and therefore 
\begin{equation}
\label{martingale}
     S(\Delta_n\|\Gamma_{\Delta_n}) \ + \ S(\Gamma_{\Delta_n}\| \Gamma_n ) \ = \ S(\Delta_n \|\Gamma_n  )
\end{equation}
by (\ref{superprop-equation}), as proven above.  
The norm closure of  $\bigcup \BB_n$ is equal to the CAR algebra \cite[Theorem 6.6]{AlickiFannes} in its Fock representation as a subalgebra of $\BB(\FFF(\HH))$.  The bi-commutant of $\bigcup \BB_n$, which equals that of its closure, is therefore all of $\BB(\FFF(\HH))$.  The ``filtration" $(\BB_n)_{n=1}^\infty$ thus satisfies the hypothesis of Cor.~5.12(iv) of Ref.~\cite{OhyaPetz}, and therefore the three terms in equation (\ref{martingale}) converge to $S( \Delta \| \Gamma_{\Delta} )$, $S( \Gamma_\Delta \| \Gamma )$, and $S( \Delta \| \Gamma )$ as $n \longrightarrow \infty$.  This establishes (\ref{superprop-equation}) even when $S(\Delta )$ is infinite.  
\end{proof}




\section{Appendices}
\label{Appendices}

The following nine appendices dilate on the technical background necessary for a thorough understanding of this review, and include a couple of deferred proofs. 
The appendices are titled:

\noindent \ref{Relative entropy for density operators} \quad Relative entropy for density operators

\noindent \ref{Fermion Fock spaces} \quad Fermion Fock spaces

\noindent \ref{Many-fermion states} \quad Many-fermion states

\noindent \ref{Substates of many-fermion states} \quad Substates of many-fermion states

\noindent \ref{1-particle density matrices}  \quad $1$-particle density matrices

\noindent \ref{Gauge-invariant quasi-free states}  \quad Gauge-invariant quasi-free states

\noindent \ref{Free states appendix}  \quad  Free states 

\noindent \ref{Proof of Proposition 1}  \quad  Proof of Proposition~\ref{free at last}

\noindent \ref{Proof of Lemma 1}  \quad  Proof of Lemma~\ref{good lemma}

\subsection{Relative entropy for density operators}
\label{Relative entropy for density operators}

The general definition of relative entropy for normal states on von Neumann algebras requires some modular theory \cite{OhyaPetz}.  
However, for density operators on a Hilbert space $\mathcal{X}$, which represent normal states on the von Neumann algebra of bounded operators on $\mathcal{X}$, 
a more elementary definition of relative entropy is available.

Let $A$ and $B$ denote density operators on a Hilbert space $\mathcal{X}$.
Let $\{\phi_1,\phi_2,\ldots\}$ and $\{\psi_1,\psi_2,\ldots\}$ be orthonormal bases of $\mathcal{X}$ consisting of eigenvectors of $A$ and $B$, 
respectively, with corresponding eigenvalues $p_i$ and $q_i$. 
We define 
$$\log A = \sum_{i: p_i > 0} \log(p_i) |\phi_i \rangle\!\langle \phi_i |\ ,$$   
a negative-semidefinite, but generally unbounded, operator.   
Note that $\log A$ is defined so that $\ker(\log A)=\ker ( A )$.  
The von Neumann entropy of $A$ is defined to be $ S(A) = -\Tr ( A \log A ) = -\sum p_i \log(p_i)$.  It may equal $+\infty$.  

We define $-\Tr ( A \log B )$ to be $+\infty$ if $\ker B \not\subset \ker A$,  otherwise, we define it by  
\[
     -\Tr ( A \log B ) \equals - \sum_i \sum_j | \langle \phi_i, \psi_j \rangle |^2 p_i \log q_j 
\]
as done in Ref.~\cite{ArakiLieb}.
We define the entropy of $A$ relative to $B$ by the formula
\[
     S( A \| B ) \equals \sum_i \sum_j | \langle \phi_i, \psi_j \rangle |^2 ( p_i \log p_i - p_i \log q_j + q_j - p_i) 
\]
as done in Ref.~\cite{Lindblad73}.  
The fact that the series defining $S( A \| B )$ has only nonnegative terms implies that 
\begin{equation}
\label{relative entropy}
      S( A \| B) \ = \   -  \Tr ( A \log B ) - S(A)
\end{equation}
if $S(A) < \infty$, 
and that 
\begin{equation}
\label{entropy inequality}
  S( A ) \ \le\  -\Tr(A \log B)
\end{equation}
even if $S(A)$ is infinite.  When $ S(A) = \infty$ formula (\ref{relative entropy}) cannot be used and $S( A \| B )$ may still be finite.


\subsection{Fermion Fock spaces}
\label{Fermion Fock spaces}

Let $\HH$ be a Hilbert space, the $1$-particle Hilbert space.  Unit vectors in $\HH$ are called ``orbitals."

The fermion Fock space over $\HH$, which we denote $\FFF(\HH)$, 
is the Hilbert space direct sum of alternating tensor powers of the $1$-particle Hilbert space $\HH$. 
That is, 
\begin{equation}
\label{Fock space}
   \FFF(\HH) \equals \CCC \oplus \HH \oplus \wedge^2 \HH \oplus \cdots \oplus \wedge^m \HH \oplus \cdots \cdots 
\end{equation}
where $\wedge^m \HH $ denotes the $m^{th}$ exterior (alternating tensorial) power of
$\HH$.  The first component of $\FFF(\HH)$ contains a distinguished unit vector $|\Omega \rangle$ called the ``vacuum vector."

The $m$-particle Hilbert space $\wedge^m \HH$ is the completion of the span of all tensor products $h_1 \wedge h_2 \wedge \ldots \wedge h_m$, 
where $h_1,\ldots,h_m$ are any $m>0$ vectors in $\HH$.  The tensors $h_1
\wedge h_2 \wedge \ldots \wedge h_m$ are formally
multilinear in $h_1,\ldots,h_m$ and satisfy
\[
     h_1 \wedge \cdots \wedge h_i \wedge \cdots \wedge h_j \wedge \cdots \wedge h_m
     \equals - h_1 \wedge \cdots \wedge h_j \wedge \cdots \wedge h_i \wedge \cdots \wedge h_m
\]
for $1\le i < j \le m$.  
In the context of $n$-electron systems, wedge products are usually called ``Slater determinants."  
The inner product of two Slater determinants is 
\[
       \langle h_1 \wedge \cdots \wedge h_{m} ,\ h'_1 \wedge \cdots \wedge h'_m \rangle
        \equals
        \det \big( \langle h_i, h'_j \rangle \big)_{ij=1}^{\ m}\ .
\]
This extends to an inner product on the linear span of the Slater determinants, and the completion of this linear span is the Hilbert space $\wedge^m \HH$.

Let $\BB(\FFF)$ denote the space of bounded operators on $\FFF = \FFF(\HH)$.   
Creation and annihilation operators $\hat{a}^*(h)$ and $\hat{a}(h)$ on $\FFF$ may be defined for each $h \in \HH$ as in Refs.~\cite{BratteliRobinson,AlickiFannes}.  
These creation and annihilation operators (creators and annihilators) are bounded operators that satisfy the canonical anticommutation relations and generate the Fock representation of the CAR algebra.   This is the uniform-norm closure in $\BB(\FFF)$ of the algebra of polynomials in the creators and annihilators.
We shall denote the CAR algebra over $\HH$ by  $\mathfrak{A}(\HH)$ and its Fock representation as a subalgebra of  $\BB(\FFF)$ by $\pi(  \mathfrak{A}(\HH) )$.  

The Fock representation of $\mathfrak{A}(\HH)$ on $\FFF(\HH)$ is irreducible, i.e., the commutant of $\pi(  \mathfrak{A}(\HH) )$ in $\BB(\FFF)$ is trivial.  Therefore the bi-commutant of $\pi(  \mathfrak{A}(\HH) )$, in which it is weakly dense, is all of $\BB(\FFF)$.  

Given an ordered orthonormal basis $(h_1,h_2,\ldots)$ of $\HH$, one can build an orthonormal basis of $\FFF(\HH)$, called a ``Fock basis"  
or ``occupation number" basis, using the orbitals $h_i$ as ``reference" orbitals.   The Fock basis vectors represent configurations of particles in the reference orbitals and are indexed by ``occupation lists" 
\[
     \bn = \big(\bn(1),\bn(2),\bn(3),\ldots \big)
\]
such that $\sum \bn(i) < \infty$, that is, such that the total number of particles in the configuration is finite.  
The set $$ \NN \equals \Big\{ \big(\bn(1),\bn(2),\bn(3),\ldots \big) :  \bn(i) \in \{ 0,1\},\  \sum \bn(i) < \infty \Big\}$$
indexes the possible configurations of fermions in the modes $(h_1,h_2,\ldots)$.     
The occupation list ${\bf 0} = (0,0,0,\ldots)$ is the index of the vacuum vector $|\Omega \rangle$, i.e.,  
$| {\bf 0}\rangle = |\Omega \rangle$.  
For $\bn \in \NN$ with $\sum \bn(i) > 0$, define the vector 
\begin{equation}
\label{Fock basis vector}
         |\bn\rangle \equals \hat{a}^*(h_1)^{\bn(1)}\hat{a}^*(h_2)^{\bn(2)} \cdots |\Omega \rangle
\end{equation}
(since the exponents $\bn(i)$ are eventually $0$, only finitely many creators appear to the left of $|\Omega\rangle$ in this formula).  
The orthonormal set $\big\{ |\bn\rangle :  \bn \in \NN \big\}$ is an orthonormal basis of $\FFF(\HH)$.  
It is the Fock basis defined with respect to the ordered orthonormal basis $(h_1,h_2,\ldots)$ of reference orbitals.

Though we have written the occupation lists as if they are sequences, all that is really required is a well-ordering of the set of reference orbitals, to give a definite order to the creators in formula (\ref{Fock basis vector}).   
Allowing a different kind of well-ordering facilitates the description of the isomorphism (\ref{basic isomorphism}) below.


\subsection{Many-fermion states}
\label{Many-fermion states}

We are considering many-fermion states that can be represented by density operators on the fermion Fock space $\FFF$.  
In the conventional formalism, physical observables correspond to self-adjoint operators on $\FFF$ and states correspond to 
certain linear functionals on $\BB(\FFF)$, the von Neumann algebra of bounded operators on $\FFF$.    
We are especially interested in the ``normal" states on $\BB(\FFF)$.   
A normal state on $\BB(\FFF)$ is a $\sigma$-weakly continuous linear functional 
$\omega:\BB(\FFF) \longrightarrow \CCC$ such that $\omega(I) = 1$ and $\omega(B) \ge 0$ for all positive-semidefinite $B \in \BB(\FFF)$.  
A density operator $\Delta$ on $\FFF$ describes a normal state $\omega$ on $\BB(\FFF)$ via the formula
$\omega(B) = \Tr(\Delta B) $.  Conversely, any normal state is represented in this manner by a density operator.  
 
Using the canonical anticommutation relations, polynomials in the creators and annihilators can be written as linear combinations of 
{\it normally ordered} monomials in the creators and annihilators.    
Since  $\pi(  \mathfrak{A}(\HH) )$ is $\sigma$-weakly dense in $\BB(\FFF)$, the correlations 
 \begin{equation}
 \label{correlations}
     \Tr\big(\Delta\ \hat{a}^*(f_1)\cdots \hat{a}^*(f_n)\hat{a}(g_m)\cdots \hat{a}(g_1) \big)
\end{equation}
for all $n,m \ge 0$ with $n+m>0$, and all $f_1,f_2,\ldots,f_n, g_1,\ldots,g_m \in \HH$, 
suffice to determine the density operator $\Delta$.  
That is, no other density operator can have all the same correlations (\ref{correlations}) as  $\Delta$.

A basic example of a many-fermion sate state is a Slater determinant state.  
Let $\Phi = h_1 \wedge h_2 \wedge \cdots \wedge h_n$ denote a Slater determinant vector in 
$\wedge^n \HH$, where $\{h_1 , h_2 ,\ldots, h_n\}$ is an orthonormal set in $\HH$.  
The density operator 
\begin{equation}
\label{Slater determinant state}
     0_{\CCC} \oplus \cdots \oplus 0_{\wedge^{n-1} \HH } \oplus |\Phi\rangle\!\langle \Phi |  \oplus 0_{ \wedge^{n+1} \HH } \oplus \cdots 
\end{equation}
defined relative to the decomposition (\ref{Fock space}) of $\FFF$ represents an $n$-particle ``Slater determinant state."  
We also think of the vacuum state $|\Omega \rangle\!\langle \Omega |$ as a $0$-particle Slater determinant state.


\subsection{Substates of many-fermion states}
\label{Substates of many-fermion states}

If $\HH_1$ is a closed subspace of the $1$-particle space $\HH$ and $\HH_2$ is its orthogonal complement, then the Fock space over $\HH$ is isomorphic to the tensor product of the Fock spaces over $\HH_1$ and $ \HH_2$.  
That is, if $\HH \cong \HH_1 \oplus \HH_2$, then 
\begin{equation}
\label{basic isomorphism}
   \FFF(\HH) \ \cong \ \FFF(\HH_1) \otimes \FFF(\HH_2)\ .
\end{equation}
 We shall write $\FFF_1$ for $\FFF(\HH_1)$,  $\FFF_2$ for $\FFF(\HH_2)$, and $\FFF$ for $\FFF(\HH)$

An isomorphism (\ref{basic isomorphism}) is easy to describe using Fock bases of  $\FFF_1$ and $\FFF_2$.  
Let $(f_1,f_2,\ldots)$ and $(g_1,g_2,\ldots)$ denote ordered orthonormal bases of $\HH_1$ and $\HH_2$, respectively.
Then $(f_1,f_2,\ldots,g_1,g_2,\ldots)$ is an ordered orthonormal basis of $\HH_1 \oplus \HH_2$.  
Occupation lists relative to $(f_1,f_2,\ldots,g_1,g_2,\ldots)$ 
are in one-to-one correspondence with pairs of occupation lists $(\bn_1,\bn_2) \in \NN_1 \times \NN_2$, where $\bn_1\in \NN_1$ is an occupation list relative to  $(f_1,f_2,\ldots)$ and $\bn_2\in \NN_2$ is an occupation list relative to  $(g_1,g_2,\ldots)$.  
The correspondence  
\begin{equation}
\label{isomorphism}
  |\bn\rangle \ \longleftrightarrow\   |\bn_1\rangle \otimes |\bn_2\rangle 
\end{equation}
extends to an isomorphism.

The algebra $\BB(\FFF_1)$ is isomorphic to a subalgebra of $\BB(\FFF_1 \otimes \FFF_2)  \cong \BB(\FFF)$ via the inclusion map 
$B \mapsto B \otimes I_2$, where $I_2$ denotes the identity operator on $ \FFF_2$.   
The embedding and isomorphism  
\begin{equation}
   \BB(\FFF_1) \ \hookrightarrow \  \BB(\FFF_1 \otimes \FFF_2)  \ \cong \ \BB(\FFF)\ ,
\label{embeddingandisomorphism}
\end{equation}
map the creation and annihilation operators $\hat{a}^*(f),\hat{a}(f) \in \BB(\FFF_1)$, defined for vectors $f \in \HH_1$, to the creation and annihilation operators in $\BB(\FFF)$ denoted the same way.   Let $\BB_1$ denote the isomorphic image of $\BB(\FFF_1)$ as a subalgebra of $\BB(\FFF)$.  
If $\dim(\HH_1) = d < \infty$ then $\BB_1$ is generated algebraically by the creators and annihilators and $\dim(\BB_1) = 2^d$.
If $\HH_1$ is infinite-dimensional then $\BB_1$ is the bi-commutant and weak closure of the algebra generated by the creators and annihilators pertaining to $\HH_1$.

A state  $\omega$ on the larger von Neumann algebra  $\BB(\FFF)$ induces a state on the subalgebra $\BB_1 \cong \BB(\FFF_1)$.  
We call the induced substate on  $\BB(\FFF_1)$  a ``substate" of $\omega$, the substate  ``delimited by" the orbitals in the closed subspace $\HH_1$ of $\HH$.   
It may also be called a  ``restriction" \cite{HainzlLewinSolovej} or ``localization" \cite{LewinNamRougerie} of $\omega$.

We are particularly interested in normal states on 
$
\BB(\FFF_1 \otimes \FFF_2) \cong \BB(\FFF)\ .
$
For each normal state $\omega$ there is a corresponding 
density operator $\Delta$ on $\FFF_1 \otimes \FFF_2$ such that  
$
    \omega(A) = \Tr(\Delta A)
$
for all $A \in \BB(\FFF)$.  
The induced substate $B \mapsto \omega(B \otimes I_2)$ on $ \BB(\FFF_1) $ is also normal.  
It is represented by the partial trace of $\Delta$ with respect to $\FFF_2$, i.e., by the density operator $\Delta_1$ on  $\FFF_1$ such that 
 \begin{equation}
 \label{partial trace 1}
       \Tr(\Delta_1 B) \equals \Tr_{\FFF_2}(\Delta(B \otimes I_2))
 \end{equation}
for all bounded operators $B \in \BB(\FFF_1)$.  

\subsection{1-particle density matrices}
\label{1-particle density matrices}

Consider the $n=m=1$ correlations $(\ref{correlations})$.  The map 
\begin{equation}
\label{similartothis}
(g,f) \ \longmapsto \ \Tr\big(\Delta \hat{a}^*(f)\hat{a}(g)\big),
\end{equation}
is a bounded conjugate-bilinear form, and therefore there exists a bounded operator $\gamma_\Delta$ on $\HH$ such that 
\begin{equation}
\label{defining property of $1$-pdm again}
    \langle g    | \gamma_\Delta f   \rangle \equals \Tr\big(\Delta \hat{a}^*(f)\hat{a}(g) \big)
\end{equation}
for all $f,g \in \HH$.  
We call $\gamma_\Delta$ the ``$1$-particle density matrix" or ``$1$-pdm" of $\Delta$.  
If $h \in \HH$ is any orbital,  the diagonal matrix element $ \langle h | \gamma_\Delta h  \rangle$ of the $1$-pdm is the probability that $h$ is occupied.  
The trace of $\gamma_\Delta $ is therefore the average {\it total} number of particles.  

The eigenvectors of $\gamma_{\Delta}$ are called ``natural orbitals" of $\Delta$, and the corresponding eigenvalues are the ``natural occupation numbers" of $\Delta$.  
For example, the $1$-pdm of the Slater determinant state (\ref{Slater determinant state}) is the orthogonal projector whose range is $\hbox{span}\{h_1 ,\ldots, h_n\}$.  
Thus, $n$ of the natural occupation numbers of that state are $1$ and the rest are $0$.

Let $\HH_1$ be a closed subspace of $\HH$, and let $\Delta_1$ be the substate of $\Delta$ defined in the preceding section.  
As noted there, the embedding and isomorphism (\ref{embeddingandisomorphism}) map the creation and annihilation operators 
$\hat{a}^*(f),\hat{a}(f) \in \BB(\FFF_1)$ with $f \in \HH_1$ to the creators and annihilators on $\FFF$ denoted the same way.   
Therefore, the matrix elements (\ref{defining property of $1$-pdm again}) of the $1$-pdm $ \gamma_{\Delta_1}$, defined for 
for all $f,g \in \HH_1$, are the same as the corresponding matrix elements of $ \gamma_\Delta$.  
In other words, $ \gamma_{\Delta_1}$ is the compression of $ \gamma_\Delta$ to $\HH_1 \subset \HH$.

Finally, we derive a formula for diagonal matrix elements of the $1$-pdm.
Let $(h_1,h_2,\ldots)$ be an ordered orthonormal basis of $\HH$ and define the Fock basis with reference to this system of orbitals. 
Let $\hat{a}_i^*$ and $\hat{a}_i$ denote $\hat{a}^*(h_i)$ and $\hat{a}(h_i)$, respectively.  
Using the anticommutation relations, one can verify from (\ref{Fock basis vector}) that 
\[
     \hat{a}_i^* \hat{a}_i |\bn\rangle    \equals \bn(i) |\bn\rangle  \ .
\]
Therefore 
\begin{eqnarray}
\label{property of $1$-pdm}
     \langle h_i|\gamma_{\Delta} h_i \rangle   
     & \stackrel{(\ref{defining property of $1$-pdm again})}{=} &
       \Tr( \Delta \hat{a}_i^* \hat{a}_i)  
     \equals 
     \sum\limits_{\bn \in \NN} \langle \bn | \Delta \hat{a}_i^* \hat{a}_i |\bn\rangle  
     \equals
     \sum\limits_{\bn \in \NN} \bn(i)\langle \bn | \Delta |\bn\rangle  
     \nonumber \\     
   &=&
      \sum\limits_{\bn \in \NN :\ \bn(i)=1} \langle \bn | \Delta |\bn\rangle  \ ,
\end{eqnarray} 
an expression for the probability that the $i^{th}$ reference orbital is occupied.


\subsection{Gauge-invariant quasi-free states}
\label{Gauge-invariant quasi-free states}

Recall that $\mathfrak{A}(\HH)$ denotes the (abstract) CAR algebra over a Hilbert space $\HH$ and $\pi\big( \mathfrak{A}(\HH) \big)$ 
denotes its (Fock) representation as a subalgebra of $\BB(\FFF)$.   
A state $\omega$ on $\mathfrak{A}(\HH)$ is ``quasi-free" if its $1$-particle correlations 
 $\omega\big( \hat{a}^*(f)\hat{a}(g)\big) $ and ``anomalous" correlations $\omega\big( \hat{a}(f)\hat{a}(g)\big)$ determine all of its higher correlations 
 $$\omega\big( \hat{a}^*(f_1)\cdots \hat{a}^*(f_n)a(g_m)\cdots \hat{a}(g_1)\big) $$ via Wick's formula, as in formula (2a.11) of Ref.~\cite{BachLiebSolovej}.    
The anomalous correlations of a {\it gauge-invariant} state vanish, and Wick's formula for gauge-invariant quasi-free states can be expressed compactly in terms of the state's $1$-pdm:

A state $\omega$ on $\mathfrak{A}(\HH)$ is ``gauge-invariant quasi-free"  \cite{AlickiFannes} 
 if there exists a bounded operator $Q$ on $\HH$ such that 
\begin{equation}
\label{giqf Wicks}
      \omega\big(\ \hat{a}^*(f_1)\cdots \hat{a}^*(f_n)\hat{a}(g_m)\cdots a(g_1)\ \big) 
      \equals 
        \delta_{mn}  
         \det  \big[ \langle g_i, Q f_j  \rangle \big]_{i,j=1}^n
\end{equation}
for all $f_1,\ldots,f_n,g_1,\ldots,g_m \in \HH$.   
 $Q$ is what we call the $1$-pdm of $\omega$.  Formula (\ref{giqf Wicks}) in the case $m=n=1$ implies that $Q$ has to be a positive-semidefinite contraction.   
It is known that, conversely, for any positive-semidefinite contraction $Q$ on $\HH$, there exists a unique gauge-invariant quasi-free state satisfying (\ref{giqf Wicks}).   
 
Formula (\ref{giqf Wicks}) also implies a couple of closure properties for gauge-invariant quasi-free (GIQF) states:
 \begin{trivlist}
\item{1.\quad}   If a sequence of GIQF states converges (pointwise) to a state, the limit is also GIQF.  
\item{2.\quad}  Let $\HH_1$ denote a closed subspace of $\HH$.  The CAR algebra $\mathfrak{A}(\HH_1)$ may be identified with a C$^*$-subalgebra of $\mathfrak{A}(\HH)$, and states on the latter induce states on $\mathfrak{A}(\HH_1)$ by restriction.  The restriction to $\mathfrak{A}(\HH_1)$ of a GIQF state on $\mathfrak{A}(\HH)$ is also GIQF.  
\end{trivlist}

 We are particularly interested in states represented by density operators on the Fock space $\FFF$.  
 The restriction of such a state to $\pi\big( \mathfrak{A}(\HH) \big) \subset \BB(\FFF)$ defines a state on the CAR algebra $\mathfrak{A}(\HH)$.  
 We say that a density operator  $\Gamma$ on $\FFF$, or the normal state corresponding to it, 
 is GIQF if its restriction to the CAR subalgebra of $\BB(\FFF)$ is GIQF.    
Denoting the $1$-pdm of $\Gamma$ by $\gamma_\Gamma$, the Wick relations (\ref{giqf Wicks}) for a GIQF density operator $\Gamma$ are that 
\begin{equation}
\label{Wicks}
      \Tr\big(\ \Gamma \ \hat{a}^*(f_1)\cdots \hat{a}^*(f_n)a(g_m)\cdots a(g_1) \ \big) 
      \equals 
       \delta_{mn} \det \big[\langle g_i, \gamma_\Gamma f_j  \rangle \big]_{i,j=1}^n
\end{equation}
for all $m,n$ such that $m+n>0$ and all $f_1,\ldots,f_n,g_1,\ldots,g_n \in \HH$.


\subsection{Free states}
\label{Free states appendix}

By our definition, a many-fermion state is free if it is represented by a GIQF density operator on a fermion Fock space and has finite expected particle number.  
Since substates of GIQF states are GIQF, and since substates of states of finite expected particle number also have finite expected particle number, substates of free states are free.

The $1$-pdm of a free density operator on $\FFF(\HH)$ is positive-semidefinite contraction on $\HH$ with finite trace.  
Conversely, any positive-semidefinite contraction operator on $\HH$ is the $1$-pdm of a unique free state on $\FFF(\HH)$.

\begin{proposition}  
\label{free existence}
Suppose $Q:\HH \longrightarrow \HH$  is a positive-semidefinite contraction operator with finite trace.  
Then there exists a unique free density operator on $\FFF(\HH)$ with $1$-pdm $Q$.  
\end{proposition} 

\begin{proof}
Since $Q$ is a positive-semidefinite contraction operator with finite trace, it has a spectral decomposition 
\begin{equation}
\label{1-matrix spectral}
  Q \equals \sum p_i |h_i\rangle\!\langle h_i| 
\end{equation}
where $\{h_1,h_2,\ldots\}$ is an orthonormal basis of $\HH$ consisting of eigenvectors of $Q$.  
The corresponding eigenvalues $p_i$ all lie in the interval $[0,1]$ and their sum, the trace of $Q$, is finite.  

Let  $\big\{ |\bn\rangle :  \bn \in \NN \big\}$ denote the Fock basis of $\FFF(\HH)$ defined with respect to the ordered basis $(h_1,h_2,\ldots)$ of reference orbitals, as in formula (\ref{Fock basis vector}).  Define 
\[
   \Gamma \equals \sum_{\bn \in \NN}   \Big\{ \prod_i  p_i^{\bn(i)}  (1-p_i)^{1-\bn(i)}  \Big\}|\bn \rangle\!\langle \bn |\ .
\]
The off-diagonal matrix elements of the $1$-pdm $\gamma_{\Gamma}$ with respect to the basis $(h_1,h_2,\ldots)$ are all equal to $0$.  
Using formula (\ref{property of $1$-pdm}) it is easy to show that the diagonal matrix element $\langle h_i|\gamma_{\Gamma} h_i \rangle$ equals $p_i$. 
Thus the $1$-pdm of $\Gamma$ equals $Q$.   As $\Tr (Q) = \sum p_i$ is finite, $\Gamma$ has finite average particle number. 

To show that $\Gamma$ is GIQF, we have to verify that Wick's relations are satisfied.   
It is fairly straightforward to verify that the Wick's relations (\ref{Wicks}) are satisfied when $f_1,\ldots,f_n$ and $g_1,\ldots,g_m$ all belong to the set $\{h_1,h_2,\ldots\}$.
This suffices to show that all relations (\ref{Wicks}) are satisfied.  For fixed $m$ and $n$, the left-hand and right-hand sides of (\ref{Wicks}) are bounded multilinear forms in $f_1,\ldots,f_n$ and $g_1,\ldots,g_m$.   Since these bounded multilinear forms agree when the $f$'s and $g$'s are all drawn from the same orthonormal basis of $\HH$, they must be equal.  

Thus $\Gamma$ is GIQF and its $1$-pdm $Q$ has finite trace.  This means that $\Gamma$ is a free density operator with $1$-pdm $Q$.  
No other GIQF density operator on $\FFF(\HH)$ can have the same $1$-pdm, since the $1$-pdm of a GIQF density operator determines all higher correlations via Wick's relations (\ref{Wicks}), and no other density operator on $\FFF(\HH)$ can have all the same correlations.
\end{proof}

The proof of the preceding proposition can be modified to prove the following characterization of free states:  

\begin{corollary}  
\label{structural}
A density operator $\Gamma$ on the fermion Fock space $\FFF(\HH)$ is free if and only if 
there exists an ordered orthonormal basis $(h_1,h_2,\ldots)$ of $\HH$ and real numbers $p_i \in [0,1]$ with $\sum p_i < \infty$ such that 
\begin{equation}
\label{Free density}
   \Gamma \equals \sum_{\bn \in \NN}   \Big\{ \prod_i  p_i^{\bn(i)}  (1-p_i)^{1-\bn(i)}  \Big\}|\bn \rangle\!\langle \bn |
\end{equation}
when written in terms of the Fock basis vectors $|\bn \rangle$ that are indexed by the occupation numbers 
$   \bn = \big(\bn(1),\bn(2),\bn(3),\ldots \big)  $ of the reference orbitals in $(h_1,h_2,\ldots)$. 
\end{corollary} 

Corollary~\ref{structural} provides us with a convenient structural formula for free states that we will use repeatedly in the sequel.  
The reference orbitals $h_i$ of the free density operator defined by formula (\ref{Free density}) are its natural orbitals, and the $p_i$ are its natural occupation numbers. 

Formula  (\ref{Free density}) shows clearly that, in a free state, the natural orbitals are occupied or unoccupied independently of one another.   
The free density operator (\ref{Free density}) is a mixture of Fock states $|\bn \rangle\!\langle \bn |$, and the weight assigned to the configuration $\bn$ is the probability of 
obtaining the outcome $\bn$ in a sequence of independent Bernoulli 
trials for the occupations $\bn(i)$ of reference orbitals $h_i$.   

 A calculation using (\ref{Free density}) shows that the von Neumann entropy of a free state with natural occupation numbers $p_i$ is 
\[
   S(\Gamma) \ = \ -\sum_i p_i \log p_i -\sum_i (1-p_i) \log (1-p_i) \ .  
\]


\subsection{Proof of Proposition 1}
\label{Proof of Proposition 1}

We recall the statement of Proposition~\ref{free at last}.

{\it 
A density operator on the fermion Fock space $\FFF(\HH)$ is free if and only if (i) its $1$-pdm has finite trace, and (ii) it is a limit in trace norm of density operators that represent substates of Slater determinant states.
}

\begin{proof}

Slater determinants states are free: they are the free states whose $1$-matrices are finite-rank orthogonal projectors.  
Substates of Slater determinant states are free, because all substates of free states are free.  
Limits of GIQF states are also GIQF.  
Therefore, any density operator $\Gamma$ on $\FFF(\HH)$ that satisfies (ii) is GIQF.  
If, in addition, the $1$-pdm of $\Gamma$ has finite trace, then $\Gamma$ free.  
This proves the sufficiency of (i) and (ii).  

To prove the necessity of conditions (i) and (ii), we show that any free state is a limit of free states whose $1$-matrices have finite rank, and that any free states whose $1$-matrix has finite rank is a substate of a Slater determinant state.  
  
Let $\Gamma$ be a free density operator with spectral representation (\ref{Free density}) and $1$-pdm (\ref{1-matrix spectral}).  
Define
\begin{equation}
\label{1-matrix compression}
  Q_N \equals \sum_{i=1}^N p_i |h_i\rangle\!\langle h_i| \ .
\end{equation}
Let $\Gamma_N$ be the unique free density with $1$-pdm $Q_N$. 
The probabilities 
$$\mathcal{P}_N(\bn) = \prod_{i=1}^N   p_i^{\bn(i)}  (1-p_j)^{1-\bn(i)}$$ 
converge for each $\bn \in \NN$ to the probabilities appearing as coefficients in (\ref{Free density}).  
Since the probability measures $\mathcal{P}_N$ converge pointwise to a probability measure on $\NN$, they converge in $\ell^1(\NN)$, and the corresponding density operators $\Gamma_N$, which are all diagonal with respect to the same Fock basis, converge in trace norm to $\Gamma$.

To conclude the proof, we show that the free density operators $\Gamma_N$ can be represented as substates of a Slater determinant states.  
We shall construct a Slater determinant $\Phi $ out of vectors in a larger Hilbert space $\HH'$, 
such that $Q_N$ is the $1$-pdm of the substate delimited by the orbitals in the subspace $\HH$.  

Let $\HH' = \HH \oplus \mathrm{span}\{ k'_1,k'_2,\ldots,k'_N\}$, where $\{ k'_1,k'_2,\ldots,k'_N\}$ is an orthonormal set of extraneous vectors, 
and define the Slater determinant $ \Phi \in \wedge^N \HH'$ by
\[
  \Phi \equals ( \sqrt{p_1} k_1 + \sqrt{1- p_1} k'_1 ) \wedge ( \sqrt{p_2} k_2 + \sqrt{1- p_2} k'_2 ) \wedge \cdots \wedge ( \sqrt{p_N} k_N + \sqrt{1- p_N} k'_N ) \ .
\]
The substate of $|\Phi \rangle\!\langle \Phi |$ delimited by the closed subspace $ \HH \cong \HH \oplus \{0\}\subset \HH'$ has $1$-pdm $Q_N$ of formula (\ref{1-matrix compression}).

Thus $\Gamma$ is a limit in trace norm of a sequence of density operators $\Gamma_N$ that represent substates of Slater determinant states.  
\end{proof}


\subsection{Proof of Lemma 1}
\label{Proof of Lemma 1}

To prove the propositions in Sec.~\ref{Special properties of nonfreeness}, we used the fact that 
\begin{equation}
\label{the fact}
   -  \Tr(\Delta \log\Gamma )  = - \Tr(\Gamma_\Delta \log\Gamma )
\end{equation}
whenever $\Gamma$ is free.  
We proved this fact only in the case where all of the natural occupation numbers of $\Gamma$ lie strictly between $0$ and $1$.  

General free states, where some of the $p_i$ may equal $0$ or $1$, are limits of Gibbs states (cf., Lemma~2.4 of Ref.~\cite{BachLiebSolovej}).  
However, we prefer to deal with free states directly, rather than as limits of Gibbs states.  
To prove formula (\ref{the fact}) in this spirit we have to keep an eye on the kernels of $\gamma_{\Delta}$ and $I-\gamma_\Gamma$.

\begin{lemma}
\label{lemma fermions}
 Let $\Gamma,\Delta \in \DD(\FFF)$ be two density operators on the fermion Fock space $\FFF$ with $1$-matrices $\gamma_{\Delta}$ and $\gamma_{\Gamma}$. 
 Suppose that $\Gamma$ is free. 
 
 The following are equivalent:
\\ \rm{(i)}\quad $\ker \Gamma \subset \ker \Delta$ 
\\ \rm{(ii)}\quad $\ker {\gamma_{\Gamma}} \subset \ker {\gamma_{\Delta}}$ and $\ker (I-\gamma_{\Gamma}) \subset \ker (I-\gamma_{\Delta})$
\end{lemma}

\begin{proof}

Consider a fermionic free density operator $\Gamma$, written as in formula (\ref{Free density}).
 Let $J_1$ denote the set of indices $i$ for which $p_i=1$.  Note that $J_1$ is a finite set, because $\sum p_i$ is assumed to be finite.  
 Let $J_0$ denote the set of indices $j$ for which $p_j=0$.  
 It may happen that $J_1 \cup J_0$ is the entire index set for the orbitals; in that case $\Gamma$ is a Slater determinant state or the vacuum state. 
Define
\begin{equation}
\label{the notation for the index set: fermions}
     \NN_{\Gamma} \equals \big\{ \bn : \  \bn(j) = 1 \hbox{ if } j \in J_1 \hbox{ and }  \bn(j) = 0 \hbox{ if } j \in J_0 \big\}\ .
\end{equation}
Let $\FFF_{\Gamma}$ denote the the closure of $\mathrm{span}\big\{ |\bn\rangle: \bn \in  \NN_{\Gamma} \big\}$, a subspace  of the fermion Fock space $\FFF(\HH)$.
Then we can see from (\ref{Free density}) that 
\begin{equation}
\label{Quasifree fermions rewritten in heaven}
   \Gamma \equals \sum_{\bn \in  \NN_{\Gamma}}   \Big\{ \prod_{j \notin J_1 \cup J_0}  p_j^{\bn(i)}  (1-p_j)^{1-\bn(j)}  \Big\}|\bn \rangle\!\langle \bn |
\end{equation}
and
\begin{equation}
\label{kernel of Gamma: fermions}
   \ker \Gamma 
   \equals	\overline{ \hbox{span} } \big\{ |\bn\rangle: \bn \notin  \NN_{\Gamma} \big\} \ .
\end{equation}

First we prove that (i) implies (ii).  

Assume that $\ker \Gamma \subset \ker \Delta$. 

If $\gamma_{\Gamma} f_j = 0$,  then  $\langle f_j | \gamma_{\Gamma}| f_j \rangle = 0$, and therefore, by (\ref{property of $1$-pdm}), 
$\langle \bn | \Gamma | \bn\rangle = 0$ for all $\bn$ such that $\bn(j) = 1$.  Therefore,  if $\bn(j) = 1$, then $| \bn\rangle \in \ker \Gamma$ and hence also  $| \bn\rangle \in \ker \Delta$.   This implies that $\langle f_j | \gamma_{\Delta}| f_j \rangle = 0$, again by (\ref{property of $1$-pdm}), and therefore $\gamma_{\Delta} f_j = 0$.
 
Similarly, if $(I-\gamma_{\Gamma} )f_j = 0$, then $\gamma_{\Gamma} f_j = f_j$ and therefore $1= \langle f_j | \gamma_{\Gamma}| f_j \rangle $.  
By (\ref{property of $1$-pdm}), $\langle \bn | \Gamma | \bn\rangle = 0$ for all $\bn$ such that $\bn(j) = 0$.  
Since  $\ker \Gamma \subset \ker \Delta$, also $\langle \bn | \Delta | \bn\rangle = 0$ for all $\bn$ such that $\bn(j) = 0$, and therefore  $\langle f_j | \gamma_{\Delta}| f_j \rangle = 1$, or $(I-\gamma_{\Delta} )f_j = 0$.

The last few paragraphs establish (ii).
Now we prove that (ii) implies (i).

Assume that $\ker {\gamma_{\Gamma}} \subset \ker {\gamma_{\Delta}}$ and $\ker (I-\gamma_{\Gamma}) \subset \ker (I-\gamma_{\Delta})$.  
We wish to prove that $\ker \Gamma \subset  \ker \Delta$.  
By (\ref{kernel of Gamma: fermions}) it suffices to show that every $|\bn\rangle$ with $ \bn \notin  \NN_{\Gamma} $ is in the kernel of $ \Delta$.
Every $ \bn \notin  \NN_{\Gamma} $ has either $\bn(j)=1$ for some $j \in J_0$, or $\bn(j)=0$ for some $j \in J_1$.  In both cases, $|\bn\rangle  \in \ker \Delta$, as we now show.

Suppose $\bn(j)=1$ for some $j \in J_0$.  Then $f_j \in \ker \gamma_{\Gamma}$ and, since $\ker {\gamma_{\Gamma}} \subset \ker {\gamma_{\Delta}}$, also $f_j \in \ker \gamma_{\Delta}$ and therefore $\langle f_j | \gamma_{\Delta}| f_j \rangle = 0$.  By (\ref{property of $1$-pdm}), 
$\langle \bn | \Delta | \bn\rangle = 0$ for all $\bn$ such that $\bn(j)=1$.  Thus, $|\bn\rangle \in \ker \Delta$ if $\bn(j)=1$ for some $j \in J_0$.  

Suppose $\bn(j)=0$ for some $j \in J_1$. Then $f_j \in \ker (I-\gamma_{\Gamma})$ and therefore, by assumption,  $f_j \in \ker (I-\gamma_{\Delta})$.  
This implies that $\gamma_{\Delta}f_j = f_j$, $\langle f_j | \gamma_{\Delta}| f_j \rangle = 1$, and $\langle \bn | \Delta | \bn\rangle = 0$ for all $\bn$ such that $\bn(j)=0$.  
 Thus, $|\bn\rangle \in \ker \Delta$ if $\bn(j)=0$ for some $j \in J_1$.  
\end{proof}

\begin{corollary}
\label{first cor}
$\ker \Gamma_\Delta \subset \ker \Delta$.  
\end{corollary}

\begin{corollary}
\label{second cor}
If $\Gamma$ is free, then $\ker \Gamma \subset \ker \Delta$ if and only if $\ker \Gamma \subset \ker  \Gamma_\Delta$.
\end{corollary}

Using these corollaries, we now complete the proof of Lemma~1.  Recall that lemma:

{\it 
Suppose $\Delta \in \DD(\FFF)$ and let $\Gamma_{\Delta}$ denote the unique free state that has the same $1$-pdm as $\Delta$.
If $\Gamma$ is free then 
\begin{equation}
\label{conclusion again}  -  \Tr(\Delta \log\Gamma )   \equals - \Tr(\Gamma_\Delta \log\Gamma )\ . \
\end{equation}
}

\begin{proof}  

Recall the notation used in formulas (\ref{the notation for the index set: fermions}) and (\ref{Quasifree fermions rewritten in heaven}).   
The unbounded operator $\log \Gamma$, defined as in Section~\ref{Relative entropy for density operators},  
is  
\begin{eqnarray*}
\lefteqn {\log \Gamma  
      \equals
        \sum_{\bn \in \NN_\Gamma}  \sum_{i \notin J_1 \cup J_0}  \Big(   \bn(i) \log(p_i) 
        + (1-\bn(i)) \log(1-p_i)   \Big) \ |  \bn \rangle\!\langle\bn  | } \\
       & = & 
       \sum_{i \notin J_1 \cup J_0}  \Bigg[ \log(p_i) \sum_{\bn \in \NN_\Gamma: \bn(i) = 1 } | \bn  \rangle \! \langle\bn  | 
        \ + \    \log(1-p_i) \sum_{\bn \in \NN_\Gamma: \bn(i) =  0 } | \bn  \rangle\!\langle\bn  | \Bigg] \ .
\end{eqnarray*}

If $\bn \notin \NN_\Gamma$ then $ | \bn\rangle \in \ker(\Gamma) \subset \ker \Delta$, 
and therefore $\langle \bn | \Delta | \bn\rangle = 0$.  
Thus, if $\ker \Gamma  \subset \ker \Delta$, then 
\[
\sum\limits_{\bn \in \NN_\Gamma:\  \bn(i) = 1 }\langle \bn | \Delta |\bn\rangle 
\equals  
\sum\limits_{\bn  \in \NN: \bn(i) = 1} \langle \bn | \Delta |\bn\rangle   
\equals  
\langle f_i|\gamma_{\Delta} | f_i \rangle  
\]
by (\ref{property of $1$-pdm}).
Using this, we have that 
\begin{eqnarray}
\lefteqn{ - \Tr(\Delta \log\Gamma ) } 
   \nonumber \\
   &= &
  - \sum_{i \notin J_1 \cup J_0}  \Bigg[ \log(p_i) \sum_{\bn \in \NN_\Gamma: \bn(i) = 1 } 
   \Tr\big( \Delta |\bn  \rangle\!\langle\bn  | \big)  
    \ + \    \log(1-p_i) \sum_{\bn \in \NN_\Gamma: \bn(i) =  0 }\Tr\big( \Delta |\bn  \rangle\!\langle\bn  | \big)  \Bigg] 
  \nonumber \\
   & = &
   - \sum_{i \notin J_1 \cup J_0}  
    \Bigg[ \log(p_i) \sum_{\bn \in \NN_\Gamma: \bn(i) = 1 } 
     \langle \bn | \Delta |\bn\rangle
       \ + \    \log(1-p_i) \sum_{\bn \in \NN_\Gamma: \bn(i) =  0 } \langle \bn | \Delta |\bn\rangle  \Bigg] 
      \nonumber \\
 & = & 
 -  \sum_{i \notin J_1 \cup J_0}  \Big[  \log(p_i)\langle f_i|\gamma_{\Delta}| f_i  \rangle 
    +  \log(1-p_i) \big(1 - \langle f_i|\gamma_{\Delta} | f_i  \rangle\big) \Big]   
\label{maximizeme fermions}
\end{eqnarray}
provided $\ker \Gamma  \subset \ker \Delta$.

Formula (\ref{maximizeme fermions}) is valid provided that $ \ker \Gamma \subset \ker \Delta$.  
By Corollary ~\ref{second cor}, if  $ \ker \Gamma \subset \ker \Delta$ then also  $\ker \Gamma \subset \ker \Gamma_{\Delta}$. 
Since $\Delta$ and $\Gamma_{\Delta}$ have the same $1$-pdm $\gamma_\Delta$, 
formula (\ref{maximizeme fermions}) implies the conclusion (\ref{conclusion again})  if  $ \ker \Gamma \subset \ker \Delta$.

If $\ker \Gamma \not\subset \ker \Delta$, then $\ker \Gamma \not\subset \ker \Gamma_\Delta$ by Corollary ~\ref{first cor}, and both $-\Tr (\Delta \log \Gamma)$ and $-\Tr (\Gamma_\Delta \log \Gamma)$ equal $+\infty$ by definition.  The conclusion (\ref{conclusion again}) holds trivially in this case.  
\end{proof}



\bigskip
\noindent {\bf Acknowledgments:}

This work has been supported by the Austrian Science Foundation
(FWF) under grant F41 (SFB ``VICOM") and W1245 (DK ``Nonlinear PDEs").
This review has benefited from insightful suggestions by anonymous referees of our unpublished manuscript \cite{GottliebMauserArchived}.





\end{document}